\documentclass[letterpaper, 10 pt, conference]{ieeeconf}  

\IEEEoverridecommandlockouts                              

\usepackage{graphicx} 
\newtheorem{theorem}{Theorem}

\newtheorem{assumption}{Assumption}
\newtheorem{lemma}[theorem]{Lemma}
\newtheorem{definition}{Definition}
\usepackage{mathtools}
\usepackage{semantic}
\usepackage{algorithm}
\usepackage{algpseudocode}
\usepackage[thinc]{esdiff}
\usepackage{xcolor}
\usepackage{amssymb}
\usepackage{bm}
\usepackage{hyperref}
\usepackage{url}

\definecolor{darkgreen}{HTML}{006400}

\newcommand{\x}{\mathbf{x}}
\newcommand{\y}{\mathbf{y}}

\newcommand{\s}{\mathbf{s}}

\title{\LARGE \bf
Balancing Fairness and Efficiency in Energy Resource Allocations}

\author{Jiayi Li, Matthew Motoki and Baosen Zhang
\thanks{The authors are with the department of Electrical and Computer Engineering at the University of Washington. Emails: \{ljy9712,mmotoki,zhangbao\}@uw.edu}
\thanks{The authors are partially supported by NSF ECCS-1942326 and NSF ECCS-2153937.}%
}

\begin{document}

\maketitle
\thispagestyle{empty}
\pagestyle{empty}

\begin{abstract}
Bringing fairness to energy resource allocation remains a challenge, due to the complexity of system structures and economic interdependencies among users and system operators' decision-making. 
The rise of distributed energy resources has introduced more diverse heterogeneous user groups, surpassing the capabilities of traditional efficiency-oriented allocation schemes.    
Without explicitly bringing fairness to user-system interaction, this disparity often leads to disproportionate payments for certain user groups due to their utility formats or group sizes. 

Our paper addresses this challenge by formalizing the problem of fair energy resource allocation and introducing the framework for aggregators. 
This framework enables optimal fairness-efficiency trade-offs by selecting appropriate objectives in a principled way. 
By jointly optimizing over the total resources to allocate and individual allocations, our approach reveals optimized allocation schemes that lie on the Pareto front, balancing fairness and efficiency in resource allocation strategies. 
\end{abstract}

\section{Introduction}\label{sec:introduction}
Distributed energy resources (DERs), such as small-scale solar and wind generation, electric vehicles, and batteries, are crucial components of the clean energy transition; they enable end-users to actively participate in the energy market by generating, storing, and potentially selling electricity back to the grid \cite{akorede2010ders}. However, individual users often cannot directly interact with the larger electricity market, facing barriers because of the complexity of energy markets, lack of economies of scale, and high transaction costs~\cite{burger2017review}. Therefore,  energy aggregators help overcome the barriers faced by individuals by negotiating for power on their behalf and distributing both the costs and benefits amongst the group.

The goal of an aggregator is typically to maximize the efficiency of the group of users. More specifically, it maximizes the total surplus (utility minus payment) of the users~\cite{sarker2016optimal,contreras2017participation,xie2022information,chen2023competitive}. A number of algorithms, both distributed and centralized~(see e.g.,~\cite{li2011optimal,carreiro2017energy,li2023socially} and the references within), have been proposed over the years.  However, focusing solely on efficiency may lead to large asymmetries in the allocation and surplus of the users. In short, the allocation can be unfair. For example, the results in ~\cite{yang2021optimal,fornier2024fairness}, as well as our own findings in this paper, demonstrate that maximizing efficiency may result in disproportionately more energy being allocated to households or businesses with a higher willingness and ability to pay, leaving fewer resources for those with lower incomes. 

The need for fairness consideration in the energy domain has been recognized and has gained importance in recent years. This highlights the need for a more comprehensive approach that balances efficiency and fairness in energy resource allocation. For example, \cite{EERE} highlights the importance of fairness by estimating the impact the Department of Energy's Office of Energy Efficiency and Renewable Energy's investments have on disadvantaged communities and minority-serving universities. In the energy justice literature, distributional justice examines the fair allocation of energy benefits and burdens~\cite{jenkins2016energy,ren2023literature}. While these studies have provided valuable insights into fairness and equity issues in energy systems, they have primarily focused on qualitative evaluations of the outcomes of particular allocation policies. There is a need for a quantitative framework that enables rigorous analysis and optimization of different allocation strategies.

This paper makes two main contributions towards this goal.  First, we formalize the problem of fair energy resource allocation, providing a framework for studying fairness in the context of energy systems. This framework allows aggregators to trace out a portion of the Pareto front and explore the optimal trade-offs between efficiency and fairness. Second, we generalize the resource allocation problem to involve jointly optimizing the total resources to allocate and the allocation to individual users. This generalization leads to new theoretical and computational challenges. In particular, the joint optimization problem is, in general, not jointly convex, which makes it difficult to solve directly using standard optimization techniques; however, we show that it can be solved effectively by searching over convex subproblems.


Our work is similar in spirit to \cite{moret2019energy}, which introduced the concept of energy collectives—a community-based market structure—that can be used to encourage fairness among market participants. However,~\cite{moret2019energy} did not explicitly model users' surplus and can lead to suboptimal tradeoffs between different fairness measures. In addition, the participants are restricted to quadratic utilities. Our approach in this paper takes a broader scope and addresses the challenges of finding optimal fairness and efficiency tradeoffs between the users.  


Fair resource allocation has been widely studied in various domains, including wireless communications~\cite{brehmer2009proportional, sinha2017incentive}, networks~\cite{chen2021bring}, and machine learning~\cite{li2021ditto}. However, fair resource allocation in the energy domain has received relatively less attention. The unique characteristics of energy systems, such as the price being used in actual payments (instead of shadow prices in communication networks), make the problem of fair energy resource allocation particularly interesting and challenging. Specifically, resource allocation problems have traditionally focused on fixed users' utilities; however, we focus on users' surpluses where the price for the resource is a function of the total resources allocated. Optimizing over surpluses allows for a more complete analysis that captures the relationship between resource allocation and pricing decisions, which is particularly relevant in the energy sector.


The remainder of the paper is structured as follows: Section~\ref{sec:problem-formulation} formalizes the fair energy resource allocation problem; Section~\ref{sec:fairness-measures} reviews the concept of $\alpha$-fairness and the price of fairness and price of efficiency metrics; Section~\ref{sec:theoretical-results} discusses our theoretical results; Section~\ref{sec:simulation-results} provides numerical results and analysis; and Section~\ref{sec:conclusion-future-work} concludes the paper and offers future research directions.

\section{Problem Formulation}\label{sec:problem-formulation}
We study the problem of allocating a group of users' surplus to reach a desirable level of fairness and efficiency. 
Consider a system involving $N$ users and a central decision maker: 
the central decision maker decides on the total resources to purchase and allocates the resources to users with the objective of maximizing the chosen system performance metric. This central decision maker is also called the aggregator and we use these two terms interchangeably in the paper. 

Denote the utility of user $i$ to be $U_i$ and the amount of energy allocated to user $i$ as $x_i$. The unit price of energy faced by the users is denoted as $p$. Given $p$, the user surplus of user $i$ is defined as 
\begin{equation}\label{eqn:surpus_def}
s_i = U_i(x_i) - p x_i.
\end{equation}
For $N$ users, we define the surplus profile $\s$ to be the vector of each user's surplus $s_i$, and the allocation profile $\x$ to be the vector of each user's allocated energy $x_i$.\footnote{Throughout the paper, we use bold to indicate vectors and matrices.}

Let $l = \sum_i x_i$ be the total energy purchased by the aggregator and let $C(l)$ be the cost of procuring this amount. As standard practice, we assume that $C$ is differentiable and set the price of energy to be to the marginal cost $p = C'(l)$~\cite{kirschen2004fundamental}. In this paper, we actually only make use of $C'(l)$. Therefore, our results apply to markets where only the price of energy is given. The total payment from the users to the aggregator, and from the aggregator to the market, is $l \cdot C'(l)$. Quadratic functions are commonly used to model costs, and these lead to prices that are affine in demand.

In this paper, we make the following assumption:
\begin{assumption} \label{assump:convexity}
Each utility function $U_i(x_i)$ is concave in $x_i$, and $U_i(0) = 0, \forall i \in \{1, \dots, N\}$. 
Furthermore, the function $C(l)$ is convex in $l$ and $C(0) = 0$.
\end{assumption} 
This assumption is very common in economic modeling for networked systems ~\cite{randall2013economic}.



To account for fairness among the users, we make use of a fairness measure, denoted as $\Phi:\mathbb{R}^N \rightarrow \mathbb{R}$. We will introduce the format of $\Phi$ in more detail in the next section. 
In the following, we first introduce the optimization problem of interest: 
\begin{equation}\label{eqn:surplus_opt}
\begin{aligned}
    \max_{\mathbf{s}} \quad & \Phi(\s)\\
    \textrm{s.t.} \quad & \s \geq 0 \\
\end{aligned}
\end{equation}
The optimization problem in~\eqref{eqn:surplus_opt} is maximizing the fairness measure over the feasible set of nonnegative surplus. We sometime use $\mathcal{S}=\{s_i \geq 0, \; \forall i\}$ to denote this feasible set. 


For the convenience of the follow-up analysis, it is easier to work with the allocations $\x$ and the load $l$, rather than directly with the surpluses. Define the feasible set of $(\x, l)$ be $\mathcal{X} = \{(\x, l) \mid \sum_{i} x_i = l, \quad \x \geq 0 , \quad \s(\x, l) \geq 0 \}$.

An interesting fact is that $\mathcal{S}$ is not necessarily convex, even for quadratic utility and quadratic cost functions (see Fig.~\ref{fig:pareto-front}). Therefore, it is not immediately clear that the problem \eqref{eqn:surplus_opt} can be solved, regardless of the fairness measure. These types of issues have been part of the reason that it is not trivial to apply results about fairness from other domains to energy.  It turns out we can solve it by directly working with allocations $\x$ and the load $l$. Define the feasible set of $(\x, l)$ be $\mathcal{X} = \{(\x, l) \mid \sum_{i} x_i = l, \quad \x \geq 0 , \quad \s(\x, l) \geq 0 \}$, and we study 
\begin{equation}\label{eqn:general_opt}
\begin{aligned}
    \max_{\x, l} \quad & \Phi(\s(\x, l))\\
    \textrm{s.t.} \quad & \sum_{i} x_i = l \\
    & \s(\x, l) \geq 0 \\
    & \x \geq 0. 
\end{aligned}
\end{equation}
We denote the optimal solution to the problem in \eqref{eqn:general_opt} as $(\mathbf{x}^{*}, l^{*})$, and use $\s^{*}$ to denote the optimal solution to \eqref{eqn:surplus_opt}. We show that \eqref{eqn:general_opt} can be efficiently solved in the next section. 

The format of $\Phi$ trades off fairness and efficiency and the central decision maker chooses $\Phi$ depending on the performance requirements of the system. In the next section, we specify $\Phi$ using the notion of $\alpha$-fairness.


\begin{figure}[ht]   
    \centering\includegraphics[width=\columnwidth]{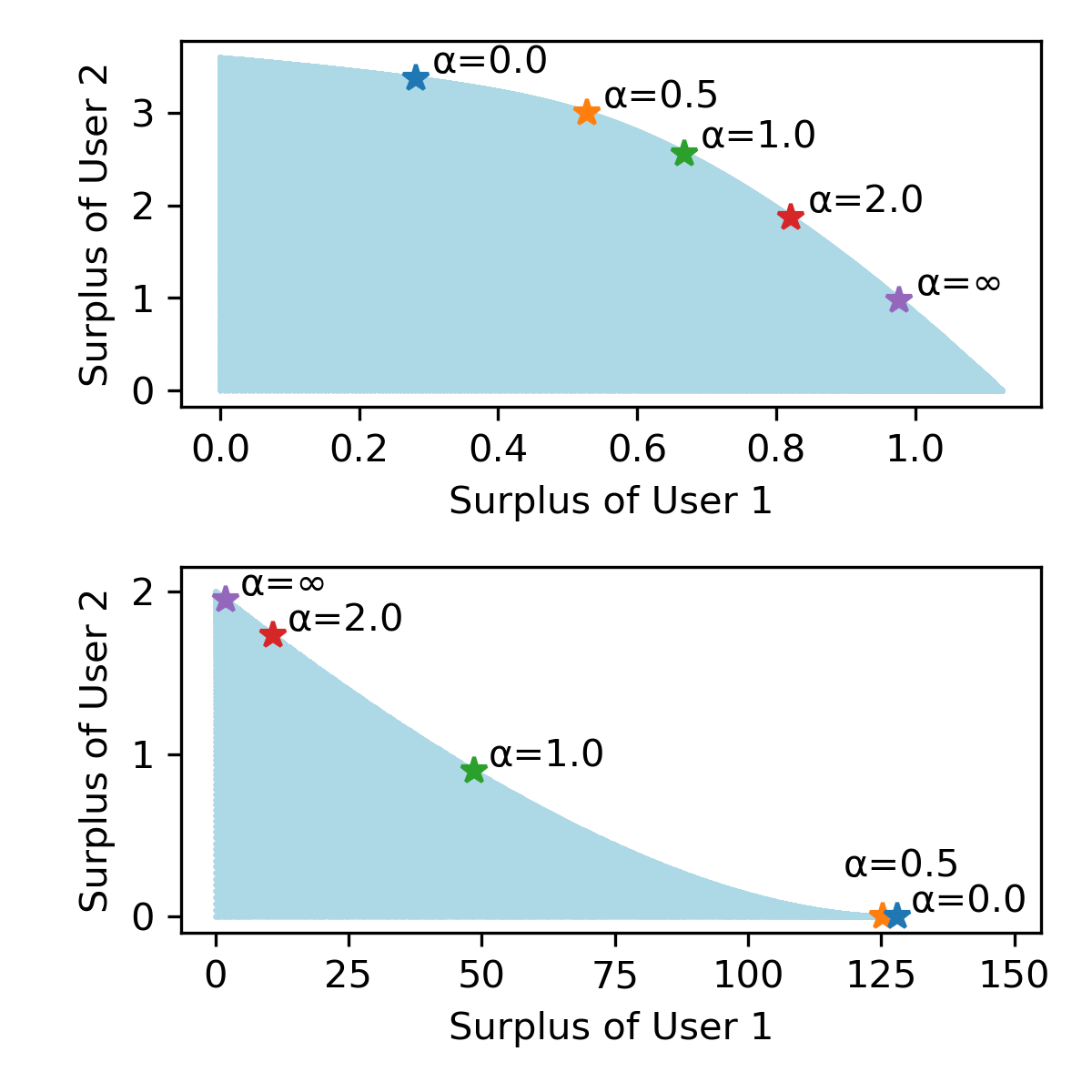}
    \caption{The figure illustrates the feasible regions and Pareto fronts for two-user systems with quadratic utilities. The top panel ($U_1(x_1) = -x_1^2 + 3x_1$ and $U_2(x_2) = -x_2^2 + 6x_2$) shows a convex feasible region, while the bottom panel ($U_1(x_1) = -x_1^2 + 40x_1$ and $U_2(x_2) = -x_2^2 + 4x_2$) shows a non-convex feasible region. Optimal $\alpha$-fairness solutions lie on the Pareto front (the upper right boundary of the feasible region) and increasing $\alpha$ traces out a portion of the Pareto front starting with the least fair social welfare solution ($\alpha=0$) to the most fair max-min solution ($\alpha=\infty$).}
    \label{fig:pareto-front}
\end{figure}



\section{Fairness Measures} \label{sec:fairness-measures}
\subsection{$\alpha$-fairness}
In this section, we detail the form of $\Phi$ that we use in this paper. 
We use the notion of $\alpha$-fairness~\cite{mo2000fair}, which provides a parametric family of functions that includes three widely used fairness measures: social welfare, proportional fairness, and max-min fairness. The idea of $\alpha$-fairness is to provide a unified framework in which the aggregator can tune the level of fairness by adjusting the $\alpha$ parameter, with higher $\alpha$ values producing more ``fair'' allocations. The $\alpha$-fairness is defined as 
\begin{equation}\label{eqn:alpha_fairness}
    \Phi(\s) = 
    \begin{cases}
        \sum_{i}\frac{s_i^{1-\alpha}}{1-\alpha} &  \text{for $\alpha \geq 0,\, \alpha \neq 1$}, \\
        \sum_{i}{\log{(s_i)}} &\text{for $\alpha=1$}.
    \end{cases}  
\end{equation}
In the following subsections, we describe the three common special cases of $\alpha$-fairness. 

\subsubsection{Social Welfare}
When $\alpha=0$, $\Phi(\s)=\sum_i s_i$, which is called the utilitarian objective, and the corresponding solutions are called the social welfare solution, denoted as $\x^{\mathcal{SW}}$. 
This solution is when the central decision-maker maximizes the sum of the surplus. This is considered to be the most ``efficient solution''~\cite{mo2000fair,kirschen2004fundamental}. 


\subsubsection{Proportional Fairness}
When $\alpha=1$, the resulting optimization problem is said to be maximizing the proportional fairness of the surpluses. It is also the generalized Nash bargaining solution for multiple users~\cite{boche2011generalization}. Proportional fairness can be intuitively understood as giving each user a proportional share of the resources based on their surpluses. It provides a compromise between efficiency and fairness, as it balances the total surplus with the individual surpluses of the users.  
Proportional fair allocation of user surplus, denoted as $\s^{\mathcal{PF}}$, should satisfy: compared to any other feasible allocation of user surplus, the aggregated proportional change is less than or equal to 0. In mathematical terms, 
\begin{equation} \label{eqn:prop_def}
    \sum_{i} \frac{s_i - s_i^{\mathcal{PF}}}{s_i^{\mathcal{PF}}} \leq 0, \forall \s \in \mathcal{S}.
\end{equation}
We state a simple lemma that shows that setting $\alpha=1$ in~\eqref{eqn:alpha_fairness} does give solutions that satisfy~\eqref{eqn:prop_def}, and the corresponding energy allocation $\x$ such that $\s(\x, \sum_i x_i)=\s^{\mathcal{PF}}$ is denoted as $\x^{\mathcal{PF}}$.
\begin{lemma}\label{proportional_fair_opt}
The proportional fair surplus profile denoted as $\s^{\mathcal{PF}}$, can be obtained as the optimal solution to the following optimization \cite{ boche2011generalization, kelly1998rate}: 

\begin{equation}\label{eqn:proportional_fair_opt}
\begin{aligned}
    \max_{\mathbf{s}} \quad & \sum_i \log(s_i)\\
    \textrm{s.t.} \quad & \s \in \mathcal{S}
\end{aligned}
\end{equation}
\end{lemma}
\begin{proof}
By first-order optimality condition, the objective of optimization problem \eqref{eqn:proportional_fair_opt} can be written as \[\Big\langle\ \mathbf{s} - \mathbf{s}^{\mathcal{PF}}, \nabla \sum_{i} \log \s_i^{\mathcal{PF}}\Big\rangle 
 = \sum_{i} \frac{s_i - s_i^{\mathcal{PF}}}{s_i^{\mathcal{PF}}}\leq 0, \forall \mathbf{s} \in \mathcal{S}\]
This inequality matches \eqref{eqn:proportional_fair_opt}. 
\end{proof}

\subsubsection{Max-min Fairness}
When $\alpha=\infty$, we obtain the max-min allocation is denoted as $\x^{\textit{MaxMin}}$. This solution maximizes the worst-case surplus for each user in the system, and is sometimes called the egalitarian solution since it is considered to be the most ``fair''.  

\subsection{Pareto Efficiency}
To signify how each different $\alpha$ leads to different fair surplus allocations, we first present the definition of Pareto optimality and Pareto front~\cite{lan2010axiomatic, joe2013multiresource}. 

\begin{definition}[Pareto Optimality]
\label{def:pareto_optimal}
A feasible surplus profile  $\mathbf{s} \in \mathcal{S}$ is Pareto optimal if there doesn't exist another set of feasible $\mathbf{\bar{s}}$ such that 
\[
s_i \leq \bar{s}_i, \forall i
\]
with at least one inequality strict.
\end{definition}
 
\begin{definition}[Pareto Front]
\label{def:pareto_front}
The Pareto front is the set of all Pareto optimal surpluses. 
\end{definition}

Pareto optimality captures the notion of optimal tradeoffs: no user can improve its surplus without decreasing other users' surpluses. Because $\Phi$ is an increasing function for all $\alpha  \geq 0$, if a surplus profile is not on the Pareto front, it is neither efficient nor fair, since there are strictly better solutions for all values of $\alpha$. 


To illustrate how different fairness notions lead to different surplus allocations within $\mathcal{S}$, as shown in Fig.~\ref{fig:pareto-front} we visualize the feasible set of surplus for an example 2-user systems with quadratic utility functions and plot the optimal surplus profile under a selected set of fair objective. 
As shown in the plots, the feasible surplus region is not always convex. Depending on the formats of users' utilities, the shape of the Pareto front, as well as the distribution of optimal $\alpha$-fairness solutions along the Pareto front, differ from one another.


Such Pareto front~\cite{jia2016dynamic} trade-off curves allow the aggregator to make informed decisions on how to balance fairness and efficiency among the users in designing an appropriate objective. It's important to note that we only explore certain regions of the Pareto front. Some points on the Pareto front are neither efficient nor fair under our chosen metric. In particular, the aggregator should only explore between utilitarian and max-min points by choosing different $\alpha$.


\subsection{Efficiency Measures}
We leverage the price of fairness (PoF) and the price of efficiency (PoE)~\cite{bertsimas2011price, bertsimas2012efficiency} to quantitatively measure the efficiency-fairness trade-offs in our systems. 

Given the feasible surplus set $\mathcal{S}$, denote the optimal utilitarian objective value as $\textit{SYSTEM}(\mathcal{S})$, the fair objective under $\Phi$ as $\textit{FAIR}(\mathcal{S}, \alpha)$; that is, $\textit{SYSTEM}(\mathcal{S})=\textit{FAIR}(\mathcal{S},0)$.
By definition, 
\[ \textit{PoF}(\mathcal{S}, \alpha) = \frac{\textit{SYSTEM}(\mathcal{S}) - \textit{FAIR}(\mathcal{S}, \alpha)}{\textit{SYSTEM}(\mathcal{S})}
\]
The price of fairness quantifies the relative decrease in total user surplus when using a fair allocation compared to the utilitarian solution. In other words, it measures the relative reduction in overall efficiency. 
Denote the Max-Min allocation as $\max_{\s \in \mathcal{S}} \min_{i} s_i$ and for each $\alpha$-fair surplus allocation, denoted as $\mathbf{z}(\alpha)$
\[ \textit{PoE}(\mathcal{S}, \alpha) = \frac{\max_{\s \in \mathcal{S}} \min_{i} s_i - \min_{i} z_i(\alpha)}{\max_{\s \in \mathcal{S}} \min_{i} s_i}
\]
The price of efficiency is the relative decrease in the minimum surplus of the users under a given allocation compared to the max-min fair allocation (the most "fair" allocation~\cite{bertsimas2012efficiency, lan2010axiomatic}). The price of efficiency measures the relative reduction in the surplus of the worst-off user.

\section{Optimization and Exploring the Pareto Front}\label{sec:theoretical-results}
\subsection{Fairness Metric and Feasible Surplus Region}


Note that for any value of $\alpha \geq 0$, the $\Phi$ function is concave and monotonically increasing~\cite{altman2008generalized}. The feasible surplus region $\mathcal{S}$ is typically assumed to be convex, thus making optimizing $\Phi$ over $\mathcal{S}$ a convex optimization problem. However, as shown in Fig.~\ref{fig:pareto-front}, $\mathcal{S}$ is not convex even for simple utility and cost functions. In the following, we work with~\eqref{eqn:general_opt} and optimize directly over $\x$ and $l$, which leads to more tractable problems. 


\subsection{Optimization Characterization}




Optimizing over $\x, l$ jointly is not a convex optimization because of the product between $C'(l)$ and $x_i$'s.  
We propose to optimize over $\x$ and $l$ separately in an iterative fashion. Given $l$, the optimization problem in \eqref{eqn:optimal_x_given_L} optimizes over $\x$:
\begin{equation}\label{eqn:optimal_x_given_L}
\begin{aligned}
    J(l)=\max_{\x} \quad & \Phi (\s(\x, l)) \\
     \textrm{s.t.} \quad & \sum_{i} x_i = l \\
     & \s(\mathbf{x}, l) \geq 0\\
     & \mathbf{x} \geq 0\\
\end{aligned}
\end{equation}
The above optimization problem is clearly convex. Because $l$ is a scalar, a grid search would find the optimal $l$ without much difficulty, as shown in Algorithm~\ref{alg:grid-search}.

\begin{algorithm}
    \caption{Grid Search: Searches through a discretized set of values for $l$ to approximate the solution for the global optimization problem \eqref{eqn:general_opt}.}
    \label{alg:grid-search}
    \begin{algorithmic}[1]
        \Require{step size $\Delta l$ and maximum value $l_\text{max}$}    
        \State {\textbf{initialize} $\phi^{*} \gets -\infty$ and $l^{*} \gets 0$}
        \For{$l$ in $[\Delta l,\, 2\Delta l,\, \dots,\, l_\text{max}]$}
            \State {$\phi \gets J(l)$}
            \If{$\phi > \phi^{*}$}
                \State {$\phi^{*} \gets s$ and $l^{*} \gets l$} 
            \EndIf
        \EndFor\\
        \Return{$l^{*}$}
    \end{algorithmic}
\end{algorithm}



        
    


\newtheorem{conjecture}[theorem]{Conjecture}

In a networked setting where $l$ is a vector, grid search could become computationally expensive. However, we make the following conjecture:
\begin{conjecture}
\label{lem:opt_on_L}
The function $J(l)$, as defined in \eqref{eqn:optimal_x_given_L}, is quasiconvex. In particular, it remains quasiconvex when $l$ is a vector, as long as the relationship between $x_i$ and $l$ is affine. 
\end{conjecture}
We empirically validated this conjecture for a large number of settings. Providing a rigorous proof is an important future direction for us. 

Last, since quadratic utility functions are commonly adopted in practice and in the literature~\cite{kirschen2004fundamental,moret2019energy}, we provide a result on when the optimization problem is jointly convex in $\x$ and $l$ under this setting.
\begin{theorem}\label{lem:quadratic_case}
If the utility functions of each user are concave and quadratic, $C$ and $C'$ are convex, and $C'$ is twice differentiable. Then the optimization problem in~\eqref{eqn:general_opt} jointly concave in $\mathbf{x}, l$ for  $\alpha=0,1,\infty$ (that is, the social welfare, proportional fair, and the max-min fairness cases).
\end{theorem}
\begin{proof}
For quadratic utility functions we have $U_i(x_i) = -a_i x_i^{2} + b_i x_i$, with $a_i > 0$. We can factor out $x_i$ from the surplus to obtain
\begin{equation*}
\begin{split}
    s_i(x_i, l) 
        &= -a_i x_i^2 + b_i x_i - C'(l) x_i \\
        &= x_i(-a_i x_i + b_i - C'(l)).
\end{split}
\end{equation*}
Thus, each non-negativity surplus constraint can be decoupled into two separate constraints $x_i \geq 0$ and  $-a_i x_i + b_i - C'(l) \geq 0$. Both of these constraints are convex, hence the feasible set is convex. Next, we look at the objective function for different values of $\alpha$. 
\begin{enumerate}
    \item When $\alpha = 0$, the objective in~\eqref{eqn:general_opt} can be written as
    \begin{equation*}
    \begin{split}
        \Phi(\s(\x, l)) 
            &= \sum_i x_i(-a_i x_i + b_i - C'(l)) \\
            &= \sum_i x_i (-a_i x_i + b_i ) - l\cdot C'(l).
    \end{split}
    \end{equation*}
    In the second line, we used $\sum_i x_i = l$. Since $\sum_i x_i (-a_i x_i + b_i)$ is concave in $\x$ by assumption, we focus on showing $l\cdot C'(l)$ is convex in $l$. As both $C$ and $C'$ are convex, we have $C''(l) \geq 0$ and $C'''(l) \geq 0$. Since we only consider $l \ge 0$, we have $ \diff[2]{(l\cdot C'(l))}{l} = 2 C''(l) + l \cdot C'''(l) \geq 0$. 


    \item When $\alpha=1$, the objective function is
    \begin{equation*}
    \begin{split}
        \Phi(\s(\x, l)) 
            &= \sum_i \log(x_i(-a_i x_i + b_i - C'(l))) \\
            &= \sum_i \log(x_i) + \log(-a_i x_i + b_i - C'(l)).
    \end{split}
    \end{equation*}
    The composition of a concave function ($\log$) and a concave function ($-a_i x_i + b_i - C'(l)$) is also concave. 
    \item When $\alpha=\infty$, we recognize that
    $$
        \max_{(\x, l) \in \mathcal{X}} \min_{i} \-a_i x_i^{2} + b_i x_i - C'(l) x_i
    $$
    is equivalent to
    $$
        \max_{(\x, l) \in \mathcal{X}} \min_{i} \log(-a_i x_i^{2} + b_i x_i - C'(l) x_i).
    $$
    As proved in previous case, $\log(-a_i x_i^{2} + b_i x_i - C'(l) x_i)$ is concave on $x_i$ and $l$. So the minimum of concave functions is a concave function on $\x, l$. 
\end{enumerate}
\end{proof}


\subsection{Pareto Efficiency}
Here we state a lemma about the Pareto optimality, which allows us to restrict our attention to the Pareto front. 



\begin{lemma} \label{def:pareto_front_def}
The optimal solution $\mathbf{s}^{*}$ to \eqref{eqn:surplus_opt} (or it's the corresponding $\mathbf{x}^{*}, l^{*}$ to \eqref{eqn:general_opt}), lies at the Pareto front of $\mathcal{S}$, which is the upper right side boundary of the set of feasible surplus. 
\end{lemma} 
\begin{proof}
A function $f:\mathbb{R}^N \rightarrow \mathbb{R}$ is component-wise strictly increasing if for $\bar{\y} \geq \y$, where the inequality is strict in at least one component,  we have $f(\bar{\y})>f(\y)$. 
The function $\Phi$ is component-wise strictly increasing for all $\alpha \geq 0$~\cite{mo2000fair}. 

Now suppose the solution to optimization~\eqref{eqn:surplus_opt}, $\s^{*}$, doesn't lie at the upper right boundary of $\mathcal{S}$. Then there exist another feasible point $\s^{*} + \bm{\epsilon}$ where $\epsilon_i >0$ for at least one $i$. Since $\Phi$ is component-wise strictly increasing, we have $ \Phi(\s^{*} + \bm{\epsilon}) > \Phi(\s^{*})$, which contradicts our assumption that $\s^{*}$ is an optimal solution to optimization problem \eqref{eqn:surplus_opt}. 
\end{proof}

\section{Simulation Results}\label{sec:simulation-results}
In this section, we demonstrate through simulations how our modeling framework could enable the aggregator to make fair allocations. We first present a simple two-user example that shows how different fairness criteria lead to different allocations and surpluses. Next, we examine how the price of fairness and price of efficiency scale with the number of users. Finally, we provide a two-class example that demonstrates how fair allocation mechanisms (specifically proportional fairness) can help reduce disparities amongst different user groups.\footnote{Our code for
the numerical simulations can be found at \url{https://github.com/lijiayi9712/fair\_resource\_allocation}.}


\subsection{Two-user example}
We now demonstrate in a simple two-user example how the social welfare solution can produce an unfair allocation while the max-min solution results in a more even allocation at the expense of efficiency and the proportionally fair solution provides a compromise between the two. In this example, the users have quadratic utilities $U_1(x_1) = -x_1^2 + 3x_1$ and $U_2(x_2) = -x_2^2 + 6x_2$. As shown in Fig.~\ref{fig:pareto-front} and Table~\ref{table:simple-example}, we see that under the social welfare solution, user 1 receives most of the allocation while user 2 receives almost nothing. On the other hand, optimizing max-min fairness results in a relatively even allocation, but the efficiency (total surplus) of the system is greatly reduced. Optimizing proportional fairness gives an allocation that is more even than the social welfare solution and has a higher efficiency than that of the max-min fairness solution.

\begin{table}[h]
    \centering
    \begin{tabular}{|l|c|c|c|c|c|}
    \hline
    & \multicolumn{2}{c|}{\textbf{Allocation}}
    & \multicolumn{3}{c|}{\textbf{Surplus}}
    \\ \cline{2-6}
    \textbf{Criterion}   & \textbf{User 1} & \textbf{User 2}  & \textbf{User 1} & \textbf{User 2}      & \textbf{Total} \\ \hline
    $\alpha=0.0$ (SW)    &  0.187  &  1.125  &  0.281  &  3.375  &  3.656 \\ \hline
    $\alpha=0.5$         &  0.427  &  0.911  &  0.527  &  3.003  &  3.530 \\ \hline
    $\alpha=1.0$ (PF)    &  0.535  &  0.682  &  0.668  &  2.564  &  3.232 \\ \hline
    $\alpha=2.0$         &  0.620  &  0.435  &  0.822  &  1.867  &  2.689 \\ \hline
    $\alpha=\infty$ (MM) &  0.691  &  0.204  &  0.977  &  0.977  &  1.954 \\ \hline
    \end{tabular}
    \caption{Allocations and surpluses for the two-user example.}
    \label{table:simple-example}
\end{table}

\subsection{Price of Fairness and Efficiency}
This section demonstrates how the PoF and PoE scale with the number of users. We run 100 experiments for various number of users. Each experiment has users with quadratic utilities $U_i(x_i)=-\frac{1}{2}a_ix_i^2 + b_ix_i$ with $a_i\sim 1+\mathrm{Unif}(0,1)$ and $b_i\sim 1 + 10\cdot(a_i+1) + 10\cdot\mathrm{Unif}(0,1)$.\footnote{To avoid experiments with small feasible regions, we chose $b_i$ to be large relative to $a_i$. This is because the constraints $x_i \ge 0$ and $\,s_i = -\frac{1}{2}a_ix_i^2 + b_ix_i - lx_i\ge 0$ imply that $0 \le x_i \le 2(b_i-l)/a_i$.} Fig.~\ref{fig:pof-poe-vs-n_users} and Fig.~\ref{fig:pof-poe-bar-chart} both show that the PoF and PoE increase as the number of users increases. As the PoF increases, the system becomes less efficient compared to the socially optimal solution. However, it is important to note that the PoF does not converge to 1, even as the number of users $N$ grows large. This implies that the efficiency loss due to fairness considerations remains bounded. Similarly, an increasing PoE indicates that the system becomes less equitable as the number of users increases. For $\alpha \neq 0$, the PoE also does not approach 1, suggesting that the fairness loss compared to the max-min fair solution is limited, even as the number of users increases.


\begin{figure}[ht] 
    \centering\includegraphics[width=\columnwidth]{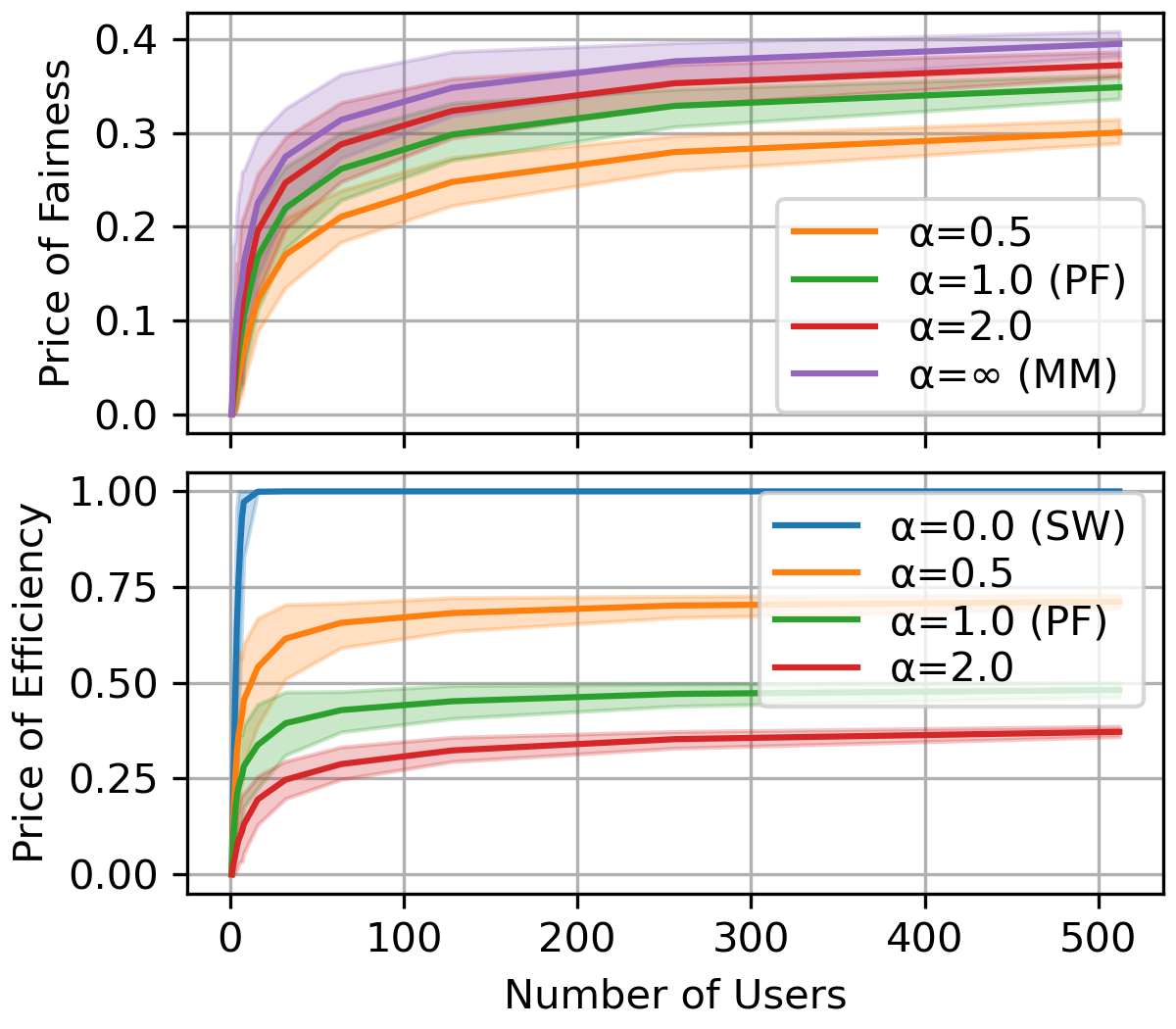}
    \caption{This graph plots the PoF and PoE for various $\alpha$-fairness criteria as a function of the number of users. The shaded areas represent the 90\% confidence interval (from the 5th to the 95th percentile) for each parameter setting. Fairness parameters closer to the socially optimal ($\alpha=0.0$) tend to have a lower PoF and higher PoE. On the other hand, fairness parameters closer to the max-min solution ($\alpha=\infty$) tend to have higher PoF and lower PoE}
    \label{fig:pof-poe-vs-n_users}
\end{figure}

\begin{figure}[ht] 
    \centering\includegraphics[width=\columnwidth]{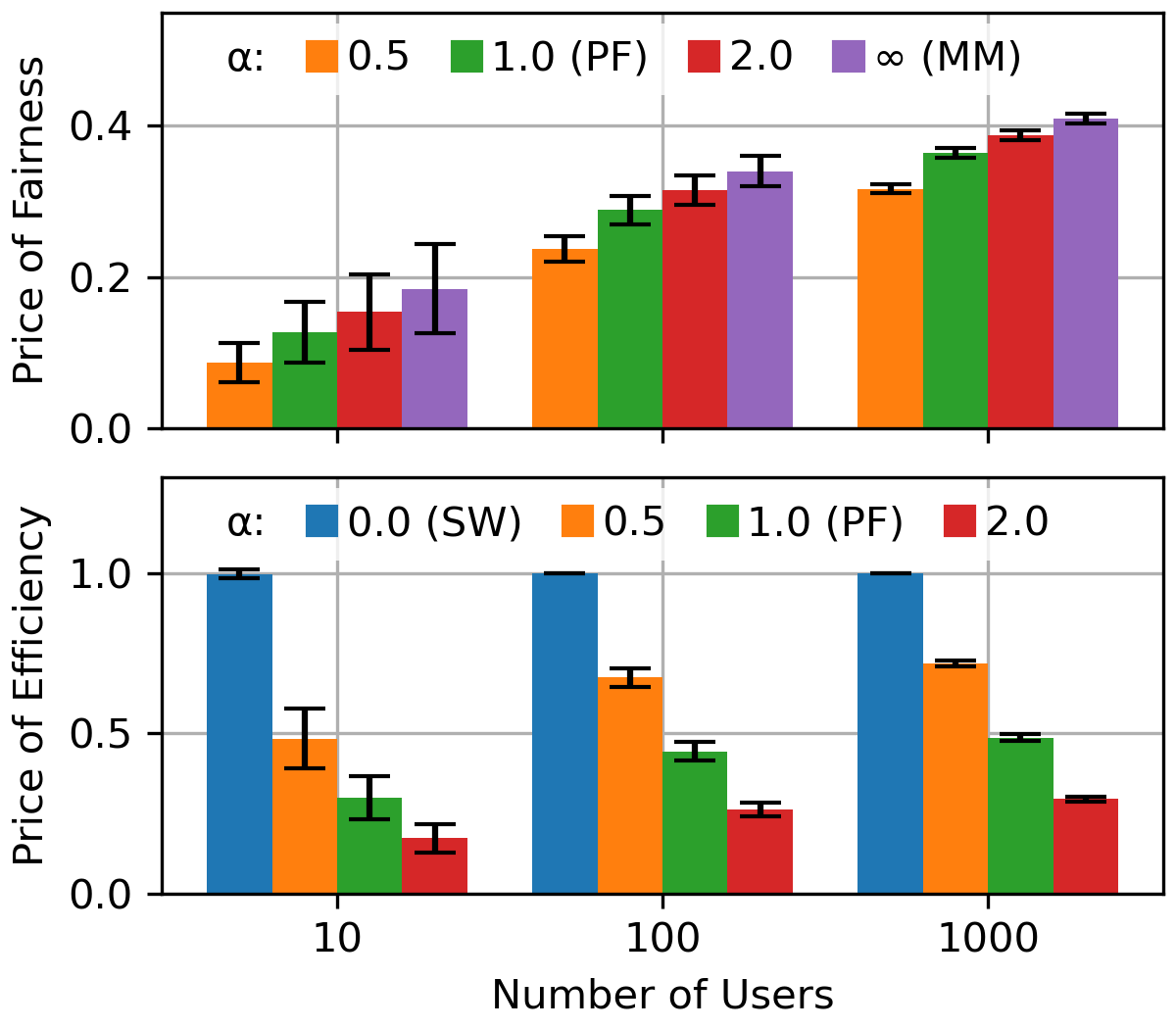}
    \caption{This graph plots the PoF and PoE for various $\alpha$-fairness criteria and number of users. Each bar represents the mean value of PoF/PoE for a specific number of users, with error bars indicating the standard deviation. Fairness parameters closer to the socially optimal ($\alpha=0.0$) tend to have a lower PoF and higher PoE. On the other hand, fairness parameters closer to the max-min solution ($\alpha=\infty$) tend to have higher PoF and lower PoE.}
    \label{fig:pof-poe-bar-chart}
\end{figure}

\subsection{Two-class Example: How Fair Objectives Help}
We now provide a simple example that splits users into two classes and demonstrates that the social welfare solution produces a large asymmetry in the allocation while the proportional fairness scheme produces allocations that are less one-sided.

To define the classes, first suppose energy was free. For quadratic utilities $U_i(x_i) = -\frac{1}{2}a_ix_i^2 + b_ix_i$, user $i$ would want to consume $x_i = b_i/a_i$ energy (set the derivative of $U_i(x_i)$ to zero and solve for $x_i$). In this example, we assume that all users have the same desired consumption when energy is free $\bar{x}=b_i/a_i$. For a non-zero price $p$, user $i$ would want to consume energy to maximize their surplus $s_i(x_i) = -\frac{1}{2}a_ix_i^2 + b_ix_i -p x_i$ or equivalently $x_i = (b_i - p) / a_i = \bar{x} - p/a_i$. Thus, the larger $a_i$ is, the more user $i$ would prefer not to deviate from $\bar{x}$. For the first class, we sample $a^{(1)}_i \sim \mathrm{Unif}(1, 2)$ and for the second class, we sample $a^{(2)}_i \sim \mathrm{Unif}(3, 4)$. 

We run 1000 experiments with 10 users in each class and compare the distribution of allocations and surpluses under the social welfare (SW) solution and the proportionally fair (PF) solution. In Fig.~\ref{fig:two-class-comparison}, we see that under the SW solution, the users in Class 1 receive almost no resources and have close to zero surpluses while the users in Class 2 receive most of the resources and large surpluses. On the other hand, the PF solution gives almost equal resources to users in both classes and the difference between the surpluses in Class 1 and Class 2 is reduced. 

We refer to the gain in allocation/surplus as how much a user's allocation/surplus changed when going from the SW solution to the PF solution. In Fig.~\ref{fig:two-class-comparison}, we see the users in Class 1 almost always benefit from moving from the SW solution to the PF solution. Some users in Class 2 also benefit from the PF solution; however, many of the users in Class 2 receive smaller allocations and lower surpluses under the PF solution.

\begin{figure}[ht] 
    \centering\includegraphics[width=\columnwidth]{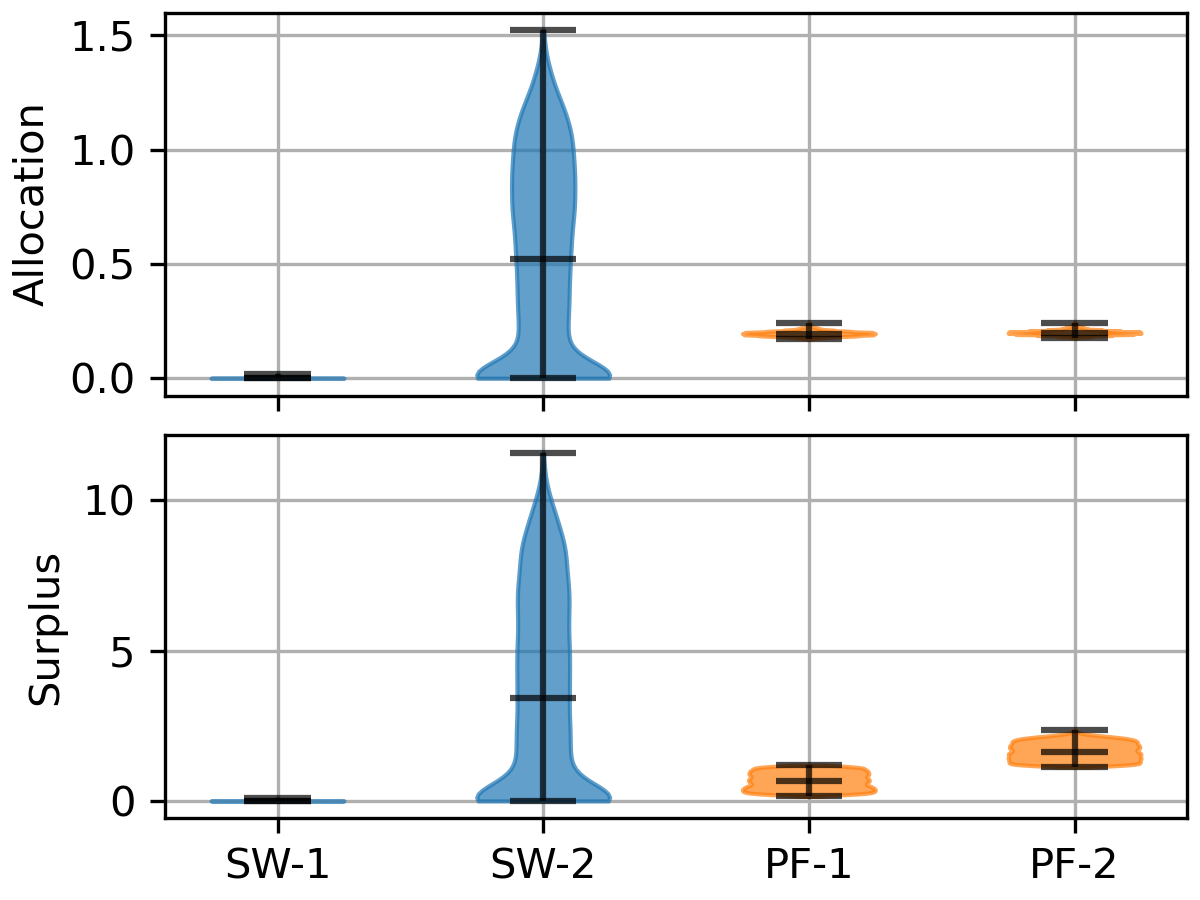}
    \caption{These plots compare the distribution of allocations (top) and surpluses (bottom) under the social welfare solution (SW) and the proportional fairness solution (PF) for Class 1 in (blue) and Class 2 (orange). The probability densities are shown in the shaded regions and the black lines indicate the minimum, median, and values.}
    \label{fig:two-class-comparison}
\end{figure}

\begin{figure}[ht] 
    \centering\includegraphics[width=\columnwidth]{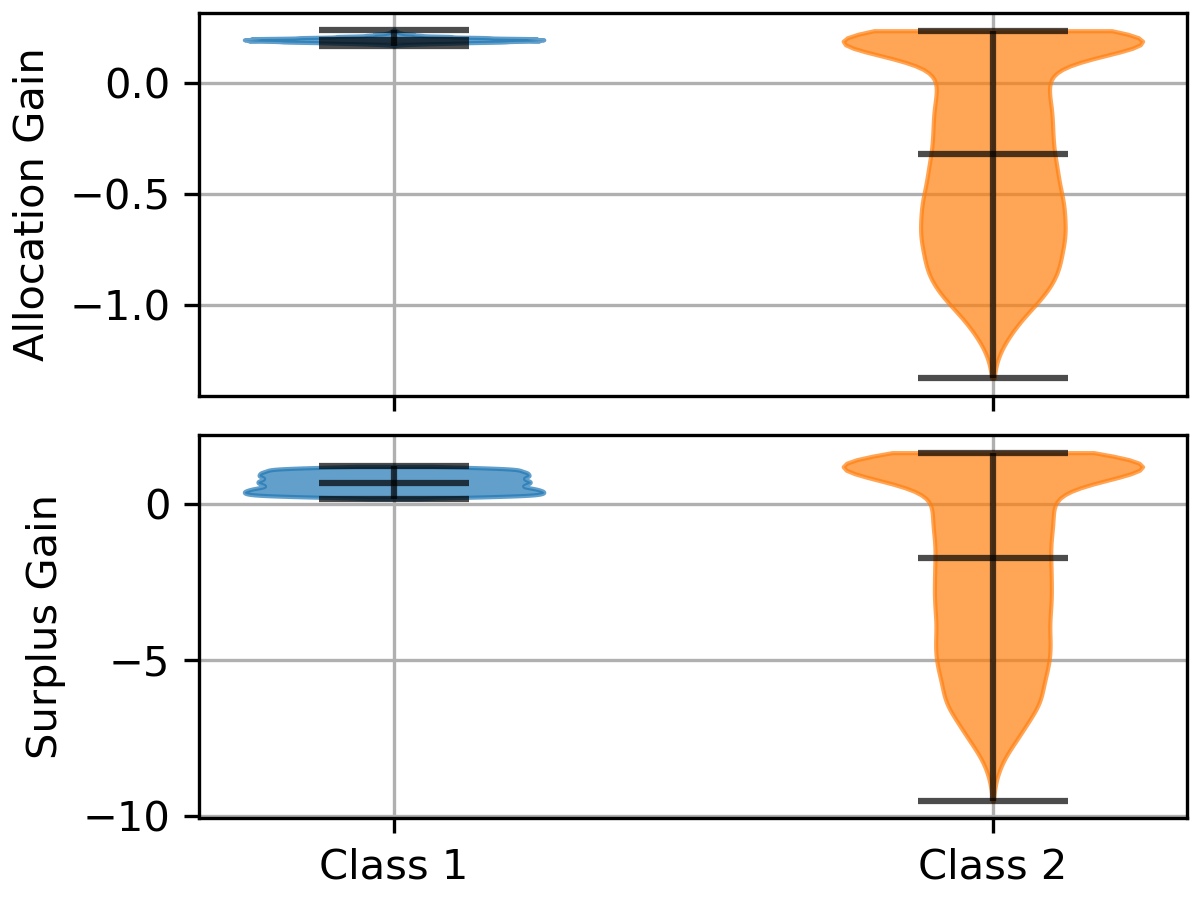}
    \caption{These plots illustrate the gains in allocation (top) and surplus (bottom) when switching from the social welfare solution to the proportional fairness solutions for Class 1 in (blue) and Class 2 (orange). Positive values indicate that the proportional fairness solution provides higher allocations or surpluses compared to the social welfare solution, while negative values show the opposite. The probability densities are shown in the shaded regions and the black lines indicate the minimum, median, and values.}
    \label{fig:two-class-gain}
\end{figure}

\addtolength{\textheight}{-3cm}   


\section{Conclusion and Future work}\label{sec:conclusion-future-work}
In this paper, we formalized the problem of fair energy resource allocation in the context of distributed energy resources (DERs) and energy aggregators. We generalized the resource allocation problem to involve jointly optimizing the total resources to allocate and the allocation itself. By doing so, we provided a principled framework that allows aggregators to explore the trade-offs between efficiency and fairness by tracing out a portion of the Pareto front. The theoretical results, numerical simulations, and analysis presented in this paper demonstrate the effectiveness of the proposed approach in achieving fair energy resource allocation.

Our work opens up several avenues for future research. In this work, we assume the aggregator knows the users' utility functions, which may be unrealistic in some scenarios. Future research could focus on developing fair resource allocation schemes for cases where the aggregator has only partial or no knowledge of each user's utility. Additionally, applying our framework to real-world datasets and exploring methods to learn utility functions from historical data could provide valuable insights and improve the practicality of our approach. Furthermore, investigating decentralized algorithms for solving the fair energy resource allocation problem could lead to more scalable and privacy-preserving solutions.






\bibliographystyle{IEEEtran} 
\bibliography{Reference}






\end{document}